\documentclass[sigconf]{aamas} 

\usepackage[ruled,linesnumbered]{algorithm2e}
\usepackage{graphicx}  
\usepackage{float}
\usepackage{subfigure}
\usepackage{seqsplit}
\usepackage{amsthm}
\usepackage{balance} 
\newtheorem{myDef}{Definition}

\setcopyright{ifaamas}
\acmConference[AAMAS '24]{Proc.\@ of the 23rd International Conference
on Autonomous Agents and Multiagent Systems (AAMAS 2024)}{May 6 -- 10, 2024}
{Auckland, New Zealand}{N.~Alechina, V.~Dignum, M.~Dastani, J.S.~Sichman (eds.)}
\copyrightyear{2024}
\acmYear{2024}
\acmDOI{}
\acmPrice{}
\acmISBN{}

\acmSubmissionID{730}

\title[AAMAS-2024 Formatting Instructions]{A Task-Driven Multi-UAV Coalition Formation Mechanism}

\author{Xinpeng Lu}
\affiliation{
	\institution{Yangzhou University}
	\city{Yangzhou}
	\country{China}}
\email{211301216@stu.yzu.edu.cn}

\author{Heng Song*}
\thanks{*Corresponding author.}
\affiliation{
	\institution{Nanjing University of Information Science \& Technology}
	\city{Nanjing}
	\country{China}}
\email{song\_heng@foxmail.com}

\author{Huailing Ma}
\affiliation{
	\institution{Yangzhou University}
	\city{Yangzhou}
	\country{China}}
\email{211301408@stu.yzu.edu.cn}

\author{Junwu Zhu*}
\affiliation{
	\institution{Yangzhou University}
	\city{Yangzhou}
	\country{China}}
\email{jwzhu@yzu.edu.cn}

\begin{abstract}
With the rapid advancement of UAV technology, the problem of UAV coalition formation has become a hotspot. Therefore, designing task-driven multi-UAV coalition formation mechanism has become a challenging problem. However, existing coalition formation mechanisms suffer from low relevance between UAVs and task requirements, resulting in overall low coalition utility and unstable coalition structures. To address these problems, this paper proposed a novel multi-UAV coalition network collaborative task completion model, considering both coalition work capacity and task-requirement relationships. This model stimulated the formation of coalitions that match task requirements by using a revenue function based on the coalition's revenue threshold. Subsequently, an algorithm for coalition formation based on marginal utility was proposed. Specifically, the algorithm utilized Shapley value to achieve fair utility distribution within the coalition, evaluated coalition values based on marginal utility preference order, and achieved stable coalition partition through a limited number of iterations. Additionally, we theoretically proved that this algorithm has Nash equilibrium solution. Finally, experimental results demonstrated that the proposed algorithm, compared to currently classical algorithms, not only forms more stable coalitions but also further enhances the overall utility of coalitions effectively.
\end{abstract}


\newcommand{\BibTeX}{\rm B\kern-.05em{\sc i\kern-.025em b}\kern-.08em\TeX}

\makeatletter
\gdef\@copyrightpermission{
	\begin{minipage}{0.3\columnwidth}
		\href{https://creativecommons.org/licenses/by/4.0/}{\includegraphics[width=0.90\textwidth]{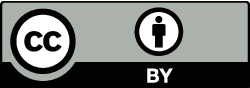}}
	\end{minipage}\hfill
	\begin{minipage}{0.7\columnwidth}
		\href{https://creativecommons.org/licenses/by/4.0/}{This work is licensed under a Creative Commons Attribution International 4.0 License.}
	\end{minipage}
	\vspace{5pt}
}
\makeatother

\begin{document}


\pagestyle{fancy}
\fancyhead{}


\maketitle 


\section{Introduction}

With the rapid advancements in technologies such as artificial intelligence and intelligent control, Unmanned Aerial Vehicles (UAVs) have witnessed extensive utilization in both military and civilian domains due to their high maneuverability, flexible control, and collective intelligence advantages \cite{gupta2015survey,shakeri2019design}. Existing UAVs typically operate individually or in coordinated formations. Compared to single UAV, which is limited by finite resource and capability, multi-UAV can be divided into multiple coalitions to jointly execute specific tasks. This approach can enhance task efficiency while also gives rise to the  multi-UAV coalition formation problem \cite{vig2006multi}. This problem entails the process where UAVs, after assessing tasks and their own attributes, form coalitions according to certain rules. Typically, a Coalition Formation Game (CFG) framework \cite{chalkiadakis2022computational} is used to address this problem, with a focus on designing mechanisms for forming coalitions and sharing utilities. Within this framework, transferable utility and coalition preferences need to be considered. Transferable utility enables the distribution of utility within the coalition, encouraging cooperation among UAVs to enhance overall utility. Coalition preferences refer to how UAVs select coalitions based on predefined rules, such as individual utility, social ranking \cite{lucchetti2022coalition}, coalition formation history \cite{boehmer2023causes}, etc.

Common coalition formation problems are typically addressed using greedy algorithms \cite{shehory1998methods}, dynamic programming \cite{rahwan2008coalition}, genetic algorithms \cite{yang2007coalition}, etc. Currently, the research mainly focus on the formation of homogeneous multi-UAV coalitions, optimizing cooperation strategies among coalition members to maximize overall utility. However, it overlooks the strong correlation between UAVs and task requirements, making it difficult to form coalitions that are highly compatible with tasks. Nonetheless, CFG framework still faces several challenges. On one hand, UAVs within the coalition have preference orders, such as Selfish order \cite{asheralieva2019hierarchical}, Pareto order \cite{zhang2018context,wang2013dynamic}. Different preference orders within the coalition lead to conflicts and insufficient cooperation among UAVs, limiting the improvement of task efficiency and collaboration ability, leaving room for enhancing the overall utility of the eventual coalition. On the other hand, existing mechanisms for utility distribution, such as equal distribution \cite{elkind2016price,2020Joint} and proportional distribution \cite{huang2022connectivity}, often fail to satisfy individual rationality, disregarding the heterogeneity of UAVs and differences in task attributes. They do not consider the varying contributions and capabilities of UAVs in tasks, reducing the incentive for active participation in coalitions and usually failing to achieve stable coalitions. To address these challenges, we research the problem of task-driven multi-UAV coalition formation and design a coalition formation algorithm based on marginal utility to achieve stable partition. The main contributions of this paper are summarized as follows:
\begin{itemize}
	\item To address the problem of mismatch between coalition formation and task requirements, we design a revenue function based on the coalition revenue threshold and propose the Motivation Incentive Theorem based on this function. This reflects the characteristics of the coalition's working capacity being related to task requirements and provides numerical basis for UAVs to join or leave a coalition.
	
	\item To resolve the problem of low utility, we define a novel preference orders based on Marginal Utility, focusing on the impact of UAVs on other UAVs upon joining a new coalition, and use the Shapley value to ensure fair utility distribution.
	
	\item To form a stable coalition structure, we propose a coalition formation algorithm based on marginal utility. We theoretically prove the convergence of this algorithm in solving Nash equilibria and show that it can achieve a stable coalition structure through a finite number of iterations. In addition, we demonstrate the effectiveness of our approach through extensive experiments on overall coalition utility.
\end{itemize}

The rest of the paper is structured as follows. Section 2 presents the related work. Section 3 provides a formal expression of the model and proves its properties. Section 4 describes the algorithm in detail. Section 5 demonstrates the effectiveness of the method through experiments. Finally, Section 6 concludes the paper.


\section{Related Work}

The formation of coalitions among multi-agent is one of the hotspots in the field of AI. In recent years, numerous scholars have conducted extensive research on this subject. In this section, we review the related work on coalition networks in multi-UAV and CFG.

\subsection{Coalition Networks in Multi-UAV}
In the context of coalition networks in multi-UAV, researchers have focused on how to achieve collaborative work among them\cite{han2023smart, tang2023swarm}. For example, Walid Saad \cite{saad2009selfish} investigated the problem of coalition formation in UAV-assisted wireless networks. However, coalitions formed solely based on selfish preferences of individual UAVs resulted in low overall utility. Xiong Fei \cite{xiong2020energy} explored the problem of data transmission in multi-UAV communication and proposed two strategies: Single Coalition Strategy (SCS) and Coalition Formation Strategy (CFS), which are based on coalition game theory. These strategies aim to optimize resource utilization and improve communication efficiency. Qi \cite{qi2022task} addressed resource allocation in UAV networks using overlapped CFG, leading to an improvement in average task utility. Nevertheless, they overlooked the strong correlation between UAVs and tasks within the approach. Jer Shyuan Ng \cite{ng2020joint} introduced a joint auction and coalition formation algorithm to tackle the allocation problem within UAV coalitions. However, due to utility-maximizing behavior exhibited by UAVs, not all coalitions formed by all participating UAVs are necessarily stable – this remains an ongoing challenge.

\subsection{Coalition Formation Games}
Teamwork, clustering, and coalition formation have always been crucial and extensively topics in computer science. For example, Alia Asheralieva \cite{asheralieva2019hierarchical} addressed the computation offloading problem in Multi-Access Edge Computing networks for multiple Service Providers using a game theory and reinforcement learning based framework. Their approach enables the formation of robust and stable MEC coalitions while providing mixed strategies for base stations under Nash equilibrium. Dolev Mutzari \cite{mutzari2021coalition} employed coalition games to research the coalition formation problem among defenders in security games and proposed a solution to compute the core of the game. Barrot \cite{barrot2019stable} explored non-enviousness and stability concepts in combination, demonstrating the existence of Pareto-optimal non-envious coalition partitions in CFG, thereby establishing a stronger theoretical foundation for this area of research. Wu \cite{wu2020monte} tackled the Coalition Structure Generation problem by proposing a Monte Carlo Tree Search-based algorithm capable of converging to an optimal solution with sufficient iterations.

\textbf{Brief conclusion}. Different from all the above scenarios about multi-UAV coalition formation problem, our work focuses on the problem related to the coalition work ability and task requirements of multi-UAV driven by tasks. By satisfying the monotonic and consistent relationship between coalition utility and task revenue, the coalition is stimulated to form a structure matching task requirements. In this paper, a coalition formation algorithm based on marginal utility is designed to improve the overall utility of the coalition while ensuring the existence of a stable coalition partition.


\section{Problem Definition}

\subsection{System Model}

In this section, we research the scenario of multi-UAV coalition for collaborative task completion. As shown in Figure~\ref{fig1}, let $\mathcal{M}=\{1,2,\dots,M\}$ represented the set of $M$ tasks and $\mathcal{N}=\{1,2,\dots,N\}$ represent the set of $N$ UAVs, with the constraint that $N\geq M$\footnote{When $N\geq M$, UAVs form coalitions to complete tasks instead of assigning multiple tasks to a single UAV.}. In this problem, each UAV $j \in \mathcal{N}$ selects to execute one of the tasks $i \in \mathcal{M}$. We define the vector $\boldsymbol{s}=(s_1,s_2,\dots,s_N)$ to denote the task selections for all UAVs, where $s_j$ represents the selected task by UAV $j$, and the vector $\boldsymbol{s}_{-j}$ represents the task selections for all other UAVs except UAV $j$. The collection of UAVs choosing to execute task $i$ denoted as $C_i \subseteq \mathcal{N} \neq \emptyset$, and we use $C_i=\{j\in \mathcal{N}|s_j=i\}$ to represent this coalition. It is required that the union of all coalitions equals the set of all UAVs i.e., $\bigcup_{i \in \mathcal{M}}C_i=\mathcal{N}$ and each UAV can only execute one specific task, i.e., $\forall k,l\in \mathcal{M}(k\neq l)$, $C_k \cap C_l=\emptyset$.
\vspace{-0.2cm}
\begin{figure}[h]
	\centering
	\includegraphics[width=1\linewidth]{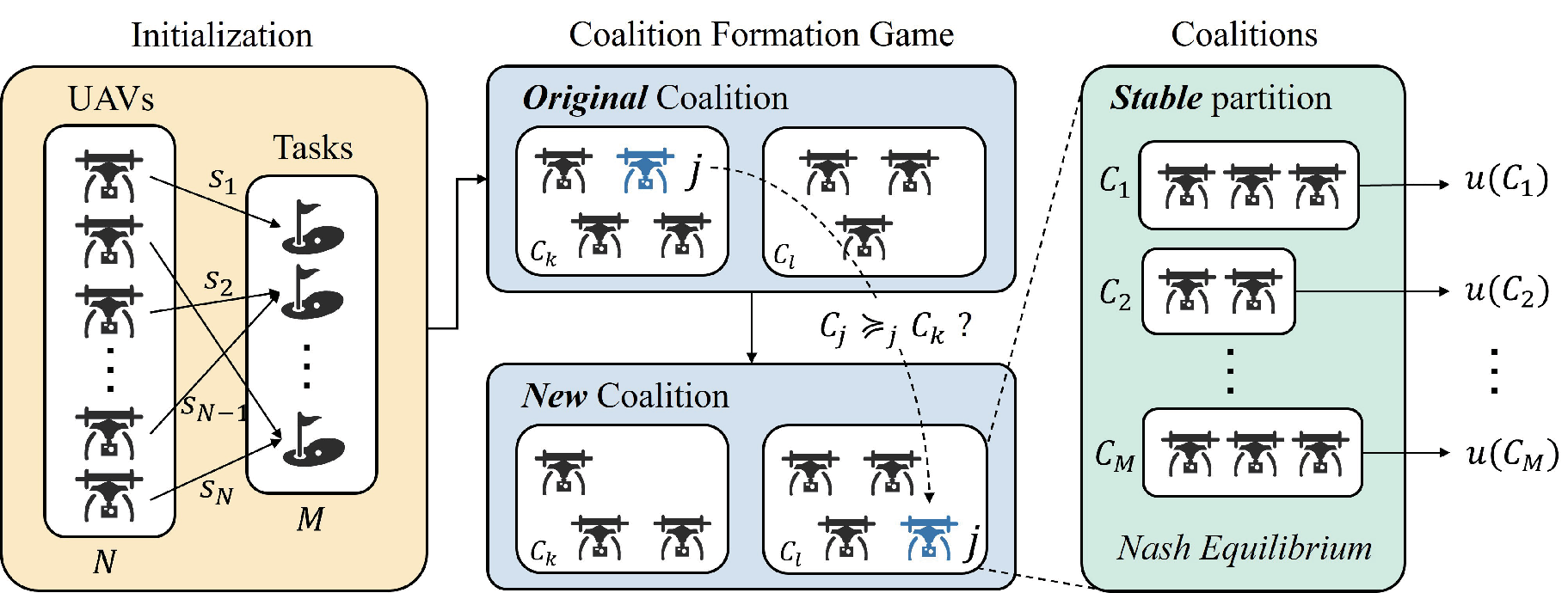}
	\caption{The schematic diagram for system model.}
	\label{fig1}
	\vspace{-0.1cm}
\end{figure}

Let the value of task $i$ be denoted as $V_i$, with a workload of $Q_i$. The efficiency of UAV $j$ in completing task $i$ is denoted as $e_j^i>0$, representing the amount of workload it can complete in one unit of time. Considering tasks such as reconnaissance and scanning that require the collaboration of multi-UAV to achieve comprehensive coverage. Within the coalition, the task is decomposed into non-overlapping subtasks, with each subtask assigned to a specific UAV based on its efficiency for completion. Therefore, the working capacity $e^i$ of the coalition $C_i$ can be expressed as the sum of the efficiencies of all UAVs, i.e., $e^i=\sum_{j \in C_i}e_j^i$. In terms of task completion time, we assume that all UAVs within the coalition start working simultaneously. Hence, the time $t_i(C_i)$ required for coalition $C_i$ to complete task $i$ is given by $t_i(C_i)=\frac{Q_i}{e^i}$.


Furthermore, we introduce the loss $L_i(C_i)$ for coalition $C_i$ to measure the consumption of completing task $i$. The calculation of loss takes into account the number of UAVs within the coalition and the task completion time. Thus, the loss $L_i(C_i)$ incurred by coalition $C_i$ in completing task $i$ is given by $L_i(C_i)=\alpha·|C_i|·t_i(C_i)$. Here, $|C_i|$ represents the number of UAVs in coalition $C_i$, and the parameter $\alpha$ represents the fixed flight cost per unit of time for UAVs, reflecting the energy consumption during UAV flight.

We denote the revenue obtained by coalition $C_i$ after completing task $i$ as $R_i(e_i)$. The coalition revenue threshold $\beta_i$ is determined based on the requirement of task $i$, specifically, representing the required working capacity of a coalition to complete task $i$. For example, tasks with a shorter expected completion time necessitate coalitions with greater working capacity. In our revenue function, the fixed task value $V_i$ represents the highest revenue that can be obtained when the coalition's capacity is equal to the task requirement. Our motivation for designing the correlation between coalition revenue and working capacity is that when there is a significant deviation between the coalition's capacity and task requirement, it can lead to excessively long completion times or resource wastage within the coalition, ultimately affecting overall task completion outcome. Specifically, inspired by the literature \cite{WU2020208}, we design two different cases for the revenue function. As shown in Figure~\ref{fig2}, these cases are described as follows:
\vspace{-0.2cm}
\begin{figure}[h]
	\centering
	\includegraphics[width=0.65\linewidth]{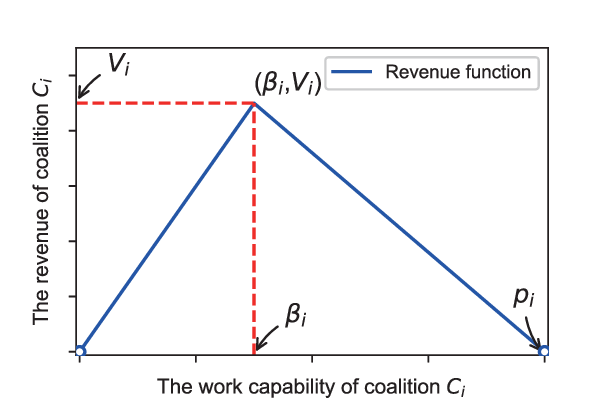}
	\caption{The schematic diagram for the revenue function based on the coalition revenue threshold.}
	\label{fig2}
\end{figure}
\vspace{-0.4cm}
\begin{itemize}
	\item When a coalition's working capacity is less than or equal to the coalition revenue threshold ($0 < e^i \leq \beta_i$), we define the task revenue as the ratio of the coalition's working capacity $e^i$ to the coalition revenue threshold $\beta_i$, i.e., $V_i e^i/\beta_i$. This represents that although a task can be completed, the potential performance of a given coalition remains unrealized due to limited working capacity resulting in low revenues.
	
	\item When a coalition's working capacity exceeds the coalition revenue threshold ($\beta_i < e^i < p_i$), the revenue derived from this excess capacity decreases linearly with the surplus, i.e., $\frac{V_i}{\beta_i-p_i}·(e^i-\beta_i)$. Here, $p_i$ represents the maximum working capacity required for task $i$, which is determined by its attributes. Since the coalition already meets the necessary working capacity for task completion, additional work capacity leads to diminishing revenue. This reflects the weakening contribution of the excess work capacity to the task.
\end{itemize}

Therefore, we can denote the revenue generated by coalition $C_i$ in completing task $i$ as 
\begin{eqnarray}\label{6}
	R_i(e^i) = \begin{cases}
		&\frac{V_i}{\beta_i} \cdot e^i, 0 < e^i \leq \beta_i \\
		&\frac{V_i}{\beta_i - p_i} \cdot (e^i - p_i), \beta_i < e^i < p_i
	\end{cases}
\end{eqnarray}

In the scenario, our goal is to enhance the utility of each coalition while forming a stable coalition structure. The utility $V_i(C_i)$ for coalition $C_i$ is denoted as $V_i(C_i)=R_i(\sum_{j\in C_i}e_j^i))-L_i(C_i)$.


To ensure that an increase (decrease) in task revenue corresponds to an increase (decrease) in coalitional utility, we establish a monotonically consistent relationship between them. This promotes coalition formation in incentive mechanisms and encourages participation in tasks that align with their capabilities.

\subsection{Proof of Theorem}

Establishing favorable model properties contributes to the design of high-performance coalition formation mechanisms. In this section, we explore and prove the desirable properties in the model proposed in Section 3.1.

\begin{theorem}
	If the coalition's utility function $V_i(C_i)$ and the task revenue function $R_i(e^i)$ exhibit the same monotonicity, then $2p_i/(1+\sqrt{1+\frac{4V_ip_i}{\alpha Q_i}}) \leq \beta_i < p_i$.
\end{theorem}

\begin{proof}
	From the model, we know that the coalition's utility function $V_i(C_i)$ for completing task $i$ is given by:
	\begin{eqnarray}
		V_i(C_i) = \begin{cases}
			&\frac{V_i}{\beta_i} \cdot e^i - \frac{\alpha Q_i |C_i|}{e^i}, 0<e^i\leq \beta_i \\
			&\frac{V_i}{\beta_i - p_i} \cdot (e^i - \beta_i) - \frac{\alpha Q_i |C_i|}{e^i}, \beta_i < e^i \leq p_i
		\end{cases}
	\end{eqnarray}
	
	If the coalition's utility function $V_i(C_i)$ and the task revenue function $R_i(e^i)$ exhibit the same monotonicity, then the coalition work capacity $e_i$ satisfies $V_i(C_i)$ increasing in the range $0<e^i \leq \beta_i$ and decreasing in the range $\beta_i < e^i \leq p_i$.
	
	\textbf{Case 1}: When $0 < e^i \leq \beta_i$, since $V_i/\beta_i>0$ and $\alpha Q_i |C_i|>0$, it follows that $V_i(C_i)$ is monotonically increasing.
	
	\textbf{Case 2}: When $\beta_i < e^i \leq p_i$, since $V_i/(\beta_i - p_i)<0$ and $\alpha Q_i |C_i|>0$, it follows that $V_i(C_i)$ is monotonically increasing within the range $\beta_i<e_i\leq \sqrt{\alpha Q_i|C_i|(p_i-\beta_i)/V_i}$ and decreasing within the range $\sqrt{\alpha Q_i|C_i|(p_i-\beta_i)/V_i} \leq e^i < p_i$. 
	
	If $V_i(C_i)$ is decreasing when $e^i>\beta_i$, then the coalition revenue threshold $\beta_i$ should meet $\beta_i\geq \sqrt{\alpha Q_i|C_i|(p_i-\beta_i)/V_i}$. Or, equivalently $\beta_i \geq 2p_i/(1+\sqrt{1+\frac{4V_ip_i}{\alpha Q_i |C_i|}})$.

	Given that the coalition revenue threshold $\beta_i$ does not exceed the maximum working capacity $p_i$ and $|C_i|\geq 1$, we can conclude $2p_i/(1+\sqrt{1+\frac{4V_ip_i}{\alpha Q_i}}) \leq \beta_i < p_i$.
\end{proof}
\vspace{-0.3cm}

Theorem 3.1 implies that, before coalition formation, each task $i$ evaluates its coalition revenue threshold $\beta_i$ based on its value $V_i$, workload $Q_i$, and maximum capacity $p_i$. This ensures that the coalition structure will meet the task requirements as much as possible.

\begin{theorem}
	\textbf{(Motivation Incentive Theorem).} When the efficiency $e_j^i$ of UAV $j$ falls within a certain controllable range, its motivation to join or leave coalition $C_i$ is stimulated.
\end{theorem}

\begin{proof}
	\textbf{Case 1}: When UAV $j$ joins coalition $C_i$, and the coalition's working capacity still does not exceed the coalition revenue threshold, i.e., $e^i+e_j^i\leq \beta_i$, the coalition utility difference is $V_i(C_i\cup\{j\})-V_i(C_i)=\frac{V_i e_j^i}{\beta_i}+\frac{\alpha Q_i(|C_i|e_j^i-e^i)}{e^i(e^i+e_j^i)}>0$. By solving the above formula, we obtain $e_j^i\geq \frac{2}{\sqrt{\lambda_1^2+\mu_1}+\lambda_1},\lambda_1=\frac{|C_i|}{e^i}+\frac{V_ie^i}{\alpha \beta_i Q_i},\mu_1=\frac{4V_i}{\alpha \beta_i Q_i}$. Therefore, when the efficiency $e_j^i$ falls within the interval $\Delta_1$:
	\begin{eqnarray}
		0<\frac{2}{\sqrt{\lambda_1^2+\mu_1}+\lambda_1} \leq e_j^i \leq \beta_i-e_j^i
	\end{eqnarray}

	UAV $j$ joining coalition $C_i$ leads to an increase in utility, that is, UAV $j$ has a positive contribution to coalition $C_i$ as $V_i(C_i\cup\{j\})-V_i(C_i)>0$. Thus, UAV $j$ has the motivation to join the coalition.
	
	\textbf{Case 2}: When UAV $j$ leaves coalition $C_i$, and the coalition's working capacity does not fall below the coalition revenue threshold, i.e., $e^i-e_j^i\geq \beta_i$, the coalition utility difference is $V_i(C_i\backslash\{j\})-V_i(C_i)=\frac{V_i e_j^i}{p_i-\beta_i}+\frac{\alpha Q_i(e^i-|C_i|e_j^i)}{e^i(e^i-e_j^i)}>0$. By solving the above formula, we obtain $e_j^i \!\leq \! \frac{2}{\lambda_2\!-\!\sqrt{\lambda_2^2+\mu_2}},\lambda_2\!=\!\frac{|C_i|}{e^i}\!-\!\frac{V^ie^i}{\alpha(p_i \!-\!\beta_i)Q_i},\mu_2=\frac{4V_i}{\alpha(p_i \!-\!\beta_i)Q_i}$. Therefore, when the efficiency $e_j^i$ falls within the interval $\Delta_2$:
	\begin{eqnarray}
		0<e_j^i \leq \min \{e^i-\beta_i,\frac{2}{\lambda_2-\sqrt{\lambda_2^2+\mu_2}}\}
	\end{eqnarray}
	UAV $j$ leaving coalition $C_i$ leads to an increase in utility, i.e., UAV $j$ has a negative contribution to coalition $C_i$ as $V_i(C_i\backslash\{j\})-V_i(C_i)>0$. Thus, UAV $j$ has the motivation to leave the coalition.
\end{proof}

Theorem 3.2 implied that when the efficiency $e_j^i$ of UAV $j$ falls within specific ranges, its motivation to join or leave coalition $C_i$ can be effectively stimulated. This helps to better understand the decision-making process of UAV coalition formation.


\section{Methodology}
\subsection{Basic Concepts}
Under the framework of CFG, a common practice for participants to increase their individual or overall utilities is to form coalitions through cooperation. We have transformed the model into a CFG with transferrable utility, where the utility can be distributed among coalition members. First, we introduce some basic concepts.

\begin{myDef}
	\textbf{(CFG)}\cite{ANDREAS2009A} A CFG can be denoted as a pair $(\mathcal{N}, (\succeq_j)_{j \in \mathcal{N}})$, where $\succeq_j$ is a complete weak preference relation over all possible coalition.
\end{myDef}

In CFG, each UAV decides whether to join or leave a coalition based on their preference order. For example, for any given two coalitions, $C_k$ and $C_l$, and UAV $j$, $C_k \succeq_j C_l$ indicates that UAV $j$ prefers to join $C_k$ rather than $C_l$, or UAV $j$ holds the same preference between the two coalitions. Furthermore, if $C_k \succ_j C_l$, it means that UAV $j$ strictly prefers $C_k$ to $C_l$. These preference orders determine the final structure of the coalition. We define the novel preference order based on Marginal Utility as follows:

\begin{myDef}
	\textbf{(Marginal Utility Order)} For any UAV $j \in \mathcal{N}$ and any two coalitions containing $j$, $C_k, C_l \subseteq \mathcal{N}$ ($k \neq l$), we have:
	\begin{eqnarray} \label{MU}
		\begin{split}
			C_k \succeq_j C_l \Leftrightarrow &u_j(C_k)+\sum_{g\in C_k\backslash\{j\}}[u_g(C_k)-u_g(C_k\backslash\{j\})]>\\
			&u_j(C_l)+\sum_{g\in C_l\backslash\{j\}}[u_g(C_l)-u_g(C_l\backslash\{j\})] 
		\end{split}
	\end{eqnarray}
\end{myDef}

Here, $u_j(C_k)$ and $u_j(C_l)$ represent the utility of UAV $j$ in coalitions $C_k$ and $C_l$, respectively, while $u_g(C_k\backslash\{j\})$ and $u_g(C_l\backslash\{j\})$ represent the utility of any UAV $g$ in coalitions $C_k$ and $C_l$ after removing UAV $j$. In the proposed order, participants consider the marginal utility on other participants when deciding to join a new coalition.

To ensure coalition stability, fairness is taken into consideration when distributing utility based on individual differences. The common absolute egalitarianism distribution method cannot meet the individual rationality. Therefore, we consider a coalition utility distribution based on Shapley value, which is defined as follows:

\begin{myDef}
	\textbf{(Shapley Value)}\cite{1953A} For any UAV $j \in C_i$ within the coalition $C_i$, the Shapley value $u_j(C_i)$ is defined as:
	\begin{eqnarray}
		u_j(C_i)=\!\!\!\!\!\sum_{C_i^\prime\subseteq C_i \backslash \{j\}}\!\!\!\!\frac{|C_i^\prime|!\texttimes(|C_i|-|C_i^{\prime}|-1)!}{|C_i|!}[V(C_i^{\prime}\cup\{j\})-V(C_i^{\prime})]
	\end{eqnarray}
\end{myDef}

Here, $C_i^{\prime}$ represents a sub-coalition of coalition $C_i$, $V(C_i^{\prime})$ represents the utility of sub-coalition $C_i^{\prime}$, and $V(C_i^{\prime}\cup{\{j\}})$ represents the utility when UAV $j$ joins the sub-coalition $C_i^{\prime}$. The Shapley value calculates the average contribution of a UAV in various possible coalitions, thus fairly distributing the utility.

\begin{myDef}
	\textbf{(Stable Coalition Partition)}\cite{2002The} A coalition partition $\mathcal{C}$ is considered stable when no participant can change the coalition structure (task selection) to increase their utility, that is
	\begin{eqnarray}
		u_j(s_j^*,\boldsymbol{s}_{-j}) \geq u_j(s_j,\boldsymbol{s}_{-j}),\forall j \in \mathcal{N},s_j \neq s_j^*
	\end{eqnarray}
\end{myDef}

To demonstrate the existence of a stable coalition structure in the proposed model, the utility function $U_j(s_j,\boldsymbol{s}_{-j})$ for UAV $j$ in the model, according to the marginal utility order, is denoted as:
\begin{eqnarray}
	U_j(s_j,\boldsymbol{s}_{-j})=u_j(s_j,\boldsymbol{s}_{-j})+\!\!\!\!\!\!\sum_{f\in C_{\hat{s_j}}\backslash\{j\}}\!\!\!\!\! u_f(s_j,\boldsymbol{s}_{-j})+\!\!\!\!\!\!\sum_{g\in C_{s_j}\backslash\{j\}}\!\!\!\!\! u_g(s_j,\boldsymbol{s}_{-j})
\end{eqnarray}

Here, $C_{s_j}$ represents the current coalition of UAV $j$, $\hat{s_j}$ represents the task previously chosen by UAV $j$, and $C_{\hat{s_j}}$ represents the previous coalition of UAV $j$.

\begin{theorem}
	Under the Marginal Utility preference order, the CFG has a Nash equilibrium solution.
\end{theorem}

\begin{proof}
	Based on the property that there exists at least one Nash equilibrium solution in exact potential games \cite{1996Potential}, we define the potential function as:
	\begin{eqnarray}
		\Phi(s_j,\boldsymbol{s}_{-j})=\sum_{j\in \mathcal{N}}u_j(s_j,\boldsymbol{s}_{-j})
	\end{eqnarray}
	
	Let's assume that UAV $j$ changes its task selection from $s_j$ to $\check{s_j}$ and joins a new coalition $C_{\check{s_j}}$. The change in utility function $U_j(s_j,\boldsymbol{s}_{-j})$ is given by:
	\begin{eqnarray}
		\begin{split}
			&U_j(\check{s_j},\boldsymbol{s}_{-j})-U_j(s_j,\boldsymbol{s}_{-j})\\&=u_j(\check{s_j},\boldsymbol{s}_{-j})-u_j(s_j,\boldsymbol{s}_{-j})\\
			&+\sum_{f\in C_{\hat{s_j}}}[u_f(s_f,\check{\boldsymbol{s}}_{-f})-u_f(s_f,\boldsymbol{s}_{-f})]\\&+\sum_{g\in C_{s_j}}[u_g(s_g,\check{\boldsymbol{s}}_{-g})-u_g(s_g,\boldsymbol{s}_{-g})]\\&+\sum_{h\in C_{\check{s_j}}}[u_h(s_h,\check{\boldsymbol{s}}_{-h})-u_h(s_h,\boldsymbol{s}_{-h})]
		\end{split}
	\end{eqnarray}
	
	Let $\Gamma_1 = C \backslash C_{\hat{s_j}} \backslash C_{s_j}$, and when UAV $j$ selects the task $s_j$, the potential function $\Phi(s_j, \boldsymbol{s}_{-j})$ is given by
	\begin{eqnarray}
		\begin{split}
			\Phi(s_j,\boldsymbol{s}_{-j})&=u_j(s_j,\boldsymbol{s}_{-j})+\sum_{f\in C_{\hat{s_j}}}u_f(s_f,\boldsymbol{s}_{-f})\\
			&+\sum_{g\in C_{s_j}}u_g(s_g,\boldsymbol{s}_{-g})+\sum_{h\in \Gamma_1 }u_h(s_h,\boldsymbol{s}_{-h})
		\end{split}
	\end{eqnarray}
	
	Let $\Gamma_2 = C \backslash C_{\hat{s_j}} \backslash C_{\check{s_j}}$, and when UAV $j$ selects the task $\check{s_j}$, the potential function $\Phi(\check{s_j}, \boldsymbol{s}_{-j})$ is given by
	\begin{eqnarray}
		\begin{split}
			\Phi(\check{s_j},\boldsymbol{s}_{-j})&=u_j(\check{s_j},\boldsymbol{s}_{-j})+\sum_{f\in 		C_{\hat{s_j}}}u_f(s_f,\check{\boldsymbol{s}}_{-f})\\
			&+\sum_{g\in C_{\check{s_j}}}u_g(s_g,\check{\boldsymbol{s}}_{-g})+\sum_{h\in \Gamma_2}u_h(s_h,\check{\boldsymbol{s}}_{-h})
		\end{split}	
	\end{eqnarray}
	
	Therefore, the difference in potential function $\Phi(s_j, \boldsymbol{s}_{-j})$ is
	\begin{eqnarray} \label{77}
		\begin{split}
			&\Phi(\check{s_j},s_{-j})-\Phi(s_j,s_{-j})\\&=u_j(\check{s_j},\boldsymbol{s}_{-j})-u_j(s_j,s_{-j})\\&+\sum_{f\in C_{\hat{s_j}}}[u_f(s_f,\check{\boldsymbol{s}}_{-f})-u_f(s_f,\boldsymbol{s}_{-f})]\\&+\sum_{g\in C_{\check{s_j}}}u_g(s_g,\check{\boldsymbol{s}}_{-g})-\!\!\!\sum_{g^{\prime}\in C_{s_j}}u_{g^{\prime}}(s_{g^{\prime}},\boldsymbol{s}_{-g^{\prime}})\\&+\sum_{h\in \Gamma_1}u_h(s_h,\check{\boldsymbol{s}}_{-h})-\sum_{h^\prime\in \Gamma_2}u_{h^\prime}(s_{h^\prime},\boldsymbol{s}_{-h^\prime})
		\end{split} 
	\end{eqnarray}
	
	Let $\Gamma = \Gamma_1 \cap \Gamma_2$, when UAV $j$ changes its task selection, it only affects the utility of UAVs within the coalitions $C_{\hat{s_j}}, C_{s_j}$, and $C_{\check{s_j}}$, and does not affect the utility of UAVs in other coalitions (i.e., the set $\Gamma$). Therefore, we can obtain
	\begin{eqnarray}\label{8}
		\sum_{h\in \Gamma}u_h(s_h,\check{\boldsymbol{s}}_{-h})=\sum_{h\in \Gamma}u_h(s_h,\boldsymbol{s}_{-h})
	\end{eqnarray}
	
	With this, we can simplify Equation~\ref{77} as follows:
	\begin{eqnarray}\label{dd}
		\begin{split}
			&\Phi(\check{s_j},\boldsymbol{s}_{-j})-\Phi(s_j,\boldsymbol{s}_{-j})\\&=u_j(\check{s_j},\boldsymbol{s}_{-j})-u_j(s_j,\boldsymbol{s}_{-j}) \\
			&+\sum_{f\in C_{\hat{s_j}}}[u_f(s_f,\check{\boldsymbol{s}}_{-f})- u_f(s_f,\boldsymbol{s}_{-f})]
			\\&+\sum_{g\in C_{s_j}}[u_g(s_g,\check{\boldsymbol{s}}_{-g})- u_g(s_g,\boldsymbol{s}_{-g})] \\ &+ \sum_{h\in C_{\check{s_j}}}[u_h(s_h,\check{\boldsymbol{s}}_{-h})- u_h(s_h,\boldsymbol{s}_{-h})]
			\\&=U_j(\check{s_j},\boldsymbol{s}_{-j})-U_j(s_j,\boldsymbol{s}_{-j})
		\end{split}
	\end{eqnarray}
	
	Therefore, Equation~\ref{dd} satisfies the definition of an exact potential game, and $\Phi(s_j,s_{-j})$ serves as the potential function for this game. As a result, the CFG has a Nash equilibrium solution.
\end{proof}

In potential games, each UAV's change in utility due to strategy variation under the Marginal Utility preference order is reflected equally in the potential function. Under this motivation, the individual objectives of the UAVs align with the global objective.

\subsection{Multi-UAV Coalition Formation Algorithm Based on Marginal Utility}
To achieve a stable coalition partition, we propose a MUCFC-CFG algorithm, as shown in Algorithm 1. The main idea of this algorithm is based on the property of finite improvement to reach a Nash equilibrium in potential games\cite{2016Potential}. In the initialization phase, takes the attributes of each UAV and task as input, and each task is set with its coalition revenue threshold (Lines 1-4). Meanwhile, each UAV randomly selects tasks and is assigned to $M$ coalitions (Lines 5-8). In the computation phase (Lines 9-13), each coalition calculates its utility, and Algorithm 2 is used to compute the utility values assigned to each UAV within the coalition.

During the iteration process (Lines 14-29), a UAV is randomly selected to change its strategy. If a UAV is selected, it attempts to choose to join another coalition, excluding the current one. Algorithm 3 is used to calculate the utility function values for the two coalitions, and by comparing these values, the preference is determined. Consequently, the corresponding strategy is updated, and the task selection is modified. Other UAVs maintain their strategies from the previous iteration. This process continues for a finite number of steps, and the algorithm converges to a stable solution through iterative optimization.

\begin{algorithm}[h]
	\caption{MUCFC-CFG Algorithm}
	\LinesNumbered
	\KwIn{Task set $\mathcal{M}$, UAV set $\mathcal{N}$}
	\KwOut{Coalitions of UAVs $C_i$ for executing each task}
	\ForEach{Task $i \in \mathcal{M}$}{
		Task $i$'s coalition revenue threshold: $\beta_i \geq 2p_i/(1+\sqrt{1+4V_ip_i/\alpha Q_i})$ \\
		Initialize coalition: $C_i\leftarrow \emptyset$ \\
	}
	\ForEach{UAV $j \in \mathcal{N}$}{
		$i\leftarrow \text{random}(1, M)$ \\
		$C_i\leftarrow C_i\cup\{j\}$, $s_j\leftarrow i$ \\
	}
	\ForEach{Task $i \in \mathcal{M}$}{
		Calculate $V_i(C_i)$ of $C_i$ using Equation \ref{6} \\
		Calculate $u_i(C_i)$ of each UAV in $C_i$ using \textbf{Algorithm 2} \\
	}
	Current coalition partition $C_{cur}\leftarrow C_{ini}$ \\
	\Repeat{Coalition partition $\mathcal{C}$ converges to a stable partition}{
		Randomly select a UAV $j, j \in \mathcal{N}$ \\
		UAV $j$ calculates $U_j(s_j)$ using \textbf{Algorithm 3} \\
		UAV $j$ randomly selects a coalition $C_{s_j^\prime},s_j^\prime \in \mathcal{M}\backslash\{s_j\}$ \\
		\eIf{$e_j^{s_j^\prime} \in \Delta_1$ and $e_j^{s_j}\in \Delta_2$}{
			Calculate the current $U_j(s_j^\prime)$ \\
			\If{$U_j(s_j^\prime) > U_j(s_j)$ (i.e., $C_{s_j^\prime} \succ_j C_{s_j}$)}
			{$C_{cur}\leftarrow C_{cur}\backslash C_{s_j}\backslash C_{s_j^\prime}$ \\
				$C_{s_j} \leftarrow C_{s_j}\backslash\{j\}$, $C_{s_j^\prime} \leftarrow C_{s_j^\prime}\cup\{j\}$ \\
				$C_{cur}\leftarrow C_{cur}\cup C_{s_j} \cup C_{s_j^\prime}$, $s_j\leftarrow s_j^\prime$ \\
			}
		}
		{
			UAV $j$ reselects a coalition $C_{s_j^\prime},s_j^\prime \in \mathcal{M}\backslash\{s_j\}$ \\
		}
	}
	\Return $C_i, i \in \mathcal{M}$
\end{algorithm}

\subsection{A Coalition Utility Allocation Algorithm Based on Shapley Values}
To achieve fair utility allocation within a coalition, we propose a Coalition Utility Allocation Algorithm based on Shapley Values (CUAA-SV), as shown in Algorithm 2. In this algorithm, the utility values for each UAV in coalition $C_i$ are initialized to zero (Line 1). In each iteration, one UAV is selected, and its contribution to the coalition's utility is calculated (Lines 2-12). All possible subsets $C_i^\prime$ of coalition $C_i$ are traversed, and if the currently selected UAV is not in the subset, the algorithm continues to the next subset. Otherwise, the weight $\omega$ of the Shapley value is calculated, taking into account the number of UAVs in the subset and the number of UAVs not in the subset (Line 8). The utility value of the selected UAV is updated according to Equation \ref{6}, which describes the change in utility when a specific UAV joins or leaves the coalition (Line 9).

\begin{algorithm}[h]
	\caption{CUAA-SV Algorithm}
	\LinesNumbered
	\KwIn{Coalition $C_i$ and parameters for task $i$: $V_i$, $Q_i$, $p_i$, $\beta_i$}
	\KwOut{Utility $u_j(C_i)$ for each UAV $j \in C_i$ in the $C_i$}
	Initialization: $u_j(C_j)\leftarrow 0$ for all $j \in C_i$ \\
	\ForEach{$j \in C_i$}{
		\For{$s=1$ \textbf{to} $|C_i|$}{
			\ForEach{$C_i^\prime \in \text{combinations}(C_i, s)$}{
				\If{$j \notin C_i^\prime$}
				{
					Continue \\
				}
				$\omega \leftarrow (|C_i^\prime|-1)!·(|C_i|-|C_i^\prime|)!/|C_i|!$ \\
				$u_j(C_i) \leftarrow u_j(C_i) + \omega \cdot (V_i(C_i) - V_i(C_i \setminus \{j\}))$ \\
	}}}
	\Return $u_j(C_i)$
\end{algorithm}

Through this iterative process, the algorithm calculates the utility that each UAV in the coalition $C_i$ should receive, ensuring that each member receives a fair utility based on their contribution to the overall task. This provides coalition members with a fair, acceptable and motivating way of distributing utility.

\subsection{A Coalition Value Evaluation Algorithm Based on Marginal Utility Order}

To facilitate coalition switching, we propose a Coalition Value Evaluation Algorithm based on Marginal Utility Order (CVEA-MUO), as shown in Algorithm 3. This algorithm is designed for a given UAV $j$ and two different coalitions, $C_{\hat{s_j}}$ and $C_{s_j}$. Initially, variables $u_1$, $u_2$, and $u_3$ are set to zero (Line 1). Here, $u_1$ represents the utility of UAV $j$ in the current coalition. By iterating through the members of coalitions $C_{\hat{s_j}}$ and $C_{s_j}$ separately, their utility values are accumulated into $u_2$ and $u_3$ (Lines 2-10). Finally, $u_1, u_2$, and $u_3$ are summed to obtain the overall utility function value $U_j(s_j, \boldsymbol{s}_{-j})$ for UAV $j$ considering different coalition (Line 11). This value is used to assist UAV $j$ in making decisions when switching coalitions.

This algorithm provides an assessment basis for decision-making under different coalition, taking into account both the utility of the individual UAV and its cooperative effects with other members.


\begin{algorithm}[h]
	\caption{CVEA-MUO Algorithm}
	\LinesNumbered
	\KwIn{Coalition $C_{\hat{s_j}}$ and $C_{s_j}$ for UAV $j$}
	\KwOut{Utility function value $U_j(s_j,s_{-j})$ for UAV $j$}
	Initialization: $u_1,u_2,u_3\leftarrow 0$, $u_1\leftarrow u_j(s_j,s_{-j})$ \\
	\ForEach{$f \in C_{\hat{s_j}}$}{
		$u_2\leftarrow u_2+u_f(C_{\hat{s_j}})$ \\
	}
	\ForEach{$g {\in {C_{s_j}}}$ }
        {
		\If{$g=f$}
		{Continue
		}
		$u_3\leftarrow u_3+u_g(C_{s_j})$ \\}
	$U_j(s_j,s_{-j})\leftarrow u_1+u_2+u_3$ \\
	\Return $U_j(s_j,s_{-j})$
\end{algorithm}

\subsection{Convergence Proof}

In this section, we provide a convergence proof for the algorithm.

\begin{figure*}[h]
	\centering
	\includegraphics[width=0.9\linewidth]{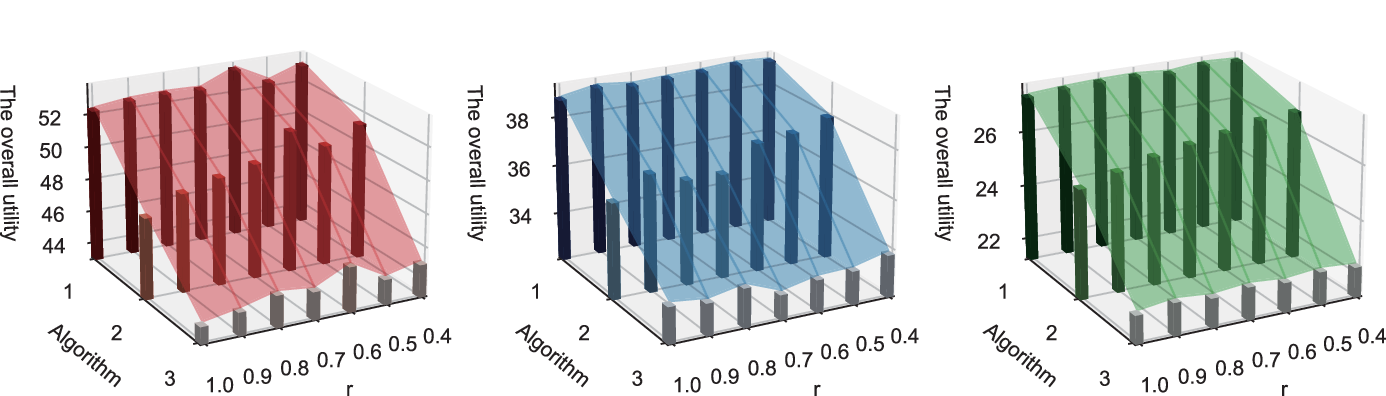}
	\caption{The variation of coalition overall utility with parameter $r$ under three scenarios.}
	\label{fig3}
\end{figure*}

\begin{theorem}
	The MUCFC-CFG algorithm converges.
\end{theorem}

\begin{proof}
	According to Definition 4.2, each UAV selects a better coalition that improves the utility function value. Therefore, any UAV updates its task selection only when the sum of the utilities within both the original and new coalitions increase. Since we have proven in Theorem 4.5 that the difference in the utility function of UAV is equal to the difference in the potential function under the marginal utility order. Consequently, the potential function increases accordingly. Note that since the number of UAVs and tasks is finite, according to the finite-increase property of potential games \cite{2016Potential}, any unilateral improvement sequence will converge to a Nash equilibrium in a finite number of iterations.
\end{proof}

\section{Experimental Results}

\subsection{Experimental Setup}

In this section, we validate the effectiveness of the proposed algorithm through experiments. The experimental code is implemented in Python\footnote{Our code is publicly available at \url{https://github.com/Agentyzu/MUCFC-CFG}, and considering pages limit, additional experiments about algorithmic robustness in large-scale scenarios are also illustrated on the Github.}. The computer's processor and memory parameters are as follows: AMD Ryzen 7 5800H with Radeon Graphics (16 CPUs) and 16GB RAM. Unless otherwise specified, the parameters for the experimental simulations are set as shown in Table~\ref{tab2}.

To better simulate real-world scenarios, we introduce the following design elements in the experimental setup:

(1) Correlation between fixed flight cost $\alpha$ and task value $V_i$: We set the $\alpha$ to be between 0.4\% and 1\% of the $V_i$ (denoted as $r=\alpha/V_i$). This considers the balance between cost and task value.

(2) Correlation between workload $Q_i$ and value $V_i$: We set the $Q_i$ to be between 1 and 1.5 of the $V_i$ (denoted as $\xi=Q_i/V_i$). Specifically, as the task value $V_i$ increases, the workload $Q_i$ also tends to increase. This because high-value tasks typically require more workload to complete.

(3) Number of UAVs and tasks: We consider scenarios of different scales, ranging from 4 to 20 UAVs, while ensuring that the number of tasks $M$ remains less than the number of UAVs to maintain the feasibility of the model.

\begin{table}[t]
	\caption{Simulation Parameters}
	\label{tab2}
	\begin{tabular}{rl}\toprule
		\textit{Parameters} & \textit{Value}\\ \midrule
	Number of tasks & $M\in[4,20]$\\
	Number of UAVs & $N\in[4,20]$\\
	Value of task $i$ & $V_i\in[5,10]$ \\
	Workload of task $i$ & $Q_i=\xi V_i$\\
	Maximum capacity of task $i$ & $p_i\in[5,6]$\\ 
	Fixed flight loss for UAVs & $\alpha=rV_i$\\ \bottomrule
\end{tabular}
\end{table}

(4) Distribution of UAV efficiency: We consider the distribution of UAV efficiency is centered around the average coalition work efficiency. This because UAV performance typically exhibits variation within a specific range.

\subsection{Experimental Results}

To demonstrate the advantages of the proposed algorithm, we compare it with Selfish and Pareto Order algorithm. Our experimental results are based on 500 rounds of randomly generated scenarios, and we take the average of the results. 

\subsubsection{Algorithm Performance}

As shown in Figure~\ref{fig3}, we compare the overall utility under three scenarios with varying numbers of UAVs and tasks ($N=20,15,10$ and $M=15,10,5$) and the different values of the parameter $r$. Under different values of $r$, our algorithm outperforms the other two in terms of the overall utility of the coalition. This is because our algorithm considers the marginal utility of each UAV through Marginal Utility Order, emphasizing the participation value of each UAV and ensuring that the inclusion of each UAV in the coalition leads to the maximum increase in overall utility of coalitions.

Additionally, our algorithm exhibits superior adaptability to variations in the parameter $r$, enabling it to maintain a relatively stable utility across different values of $r$. The other two algorithms may be more sensitive to $r$, leading to significant fluctuations in utility under different $r$. We choose to set $r$ to 0.6 and 1 in the following experiments for comparative analysis.


\begin{figure}[h]
\centering
\includegraphics[width=1\linewidth]{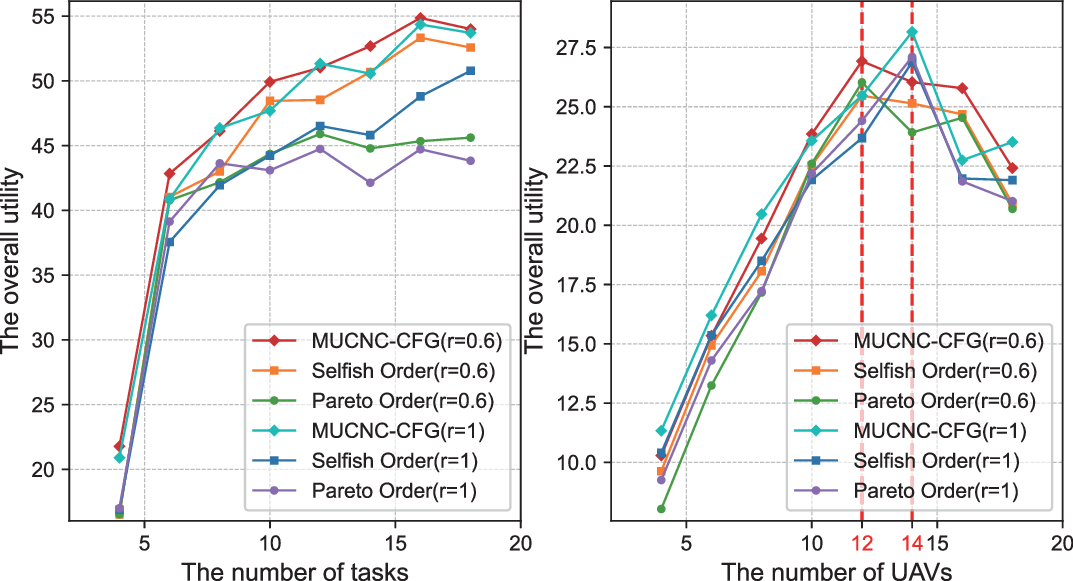}
\caption{The variation of coalition overall utility with task and UAV numbers for $N=20$ and $M=5$.}
\label{fig4}
\vspace{-0.4cm}
\end{figure}

As shown in Figure~\ref{fig4}, we compare the overall utility when the number of UAVs is $N=20$. The experimental results indicate that, when the number of UAVs is fixed, the overall utility of the coalition increases with the number of tasks increases. Meanwhile, the overall utility of the proposed algorithm and the Selfish Order grow faster as the number of tasks increases, while the utility of the Pareto Order grow lower. This is because the Selfish Order may decrease the utility of other UAVs to pursue higher individual utilities, while the Pareto Order restricts UAVs from leaving the coalition to obtain higher utility due to its strong constraints.

Additionally, we compare the overall utility when the number of tasks is $M=4$. The experimental results indicate that, when the number of tasks is fixed, the differences in utility among the three algorithms are not significant. When $r=0.6$ and $N<14$, the coalition's utility increases with the increase in the number of UAVs, reflecting the advantage of cooperation, reaching its maximum value at $N=14$. However, when more UAVs are added, the overall utility begins to decrease, indicating that the cost of cooperation outweighs the revenue. This is because, after a certain point, adding more UAVs leads to resource dispersion, and UAVs compete for tasks, thereby reducing overall utility. Therefore, when choosing the number of UAVs, a balance between cost and revenues is needed to find the most suitable number of UAVs for a specific number of tasks. It's worth noting that a larger $r$ reflects that UAVs will consume more energy. Hence, as the parameter $r$ increases, the energy consumption of UAVs also rise, resulting in a higher number of UAVs corresponding to the peak.

\begin{figure}[h]
\centering
\includegraphics[width=1\linewidth]{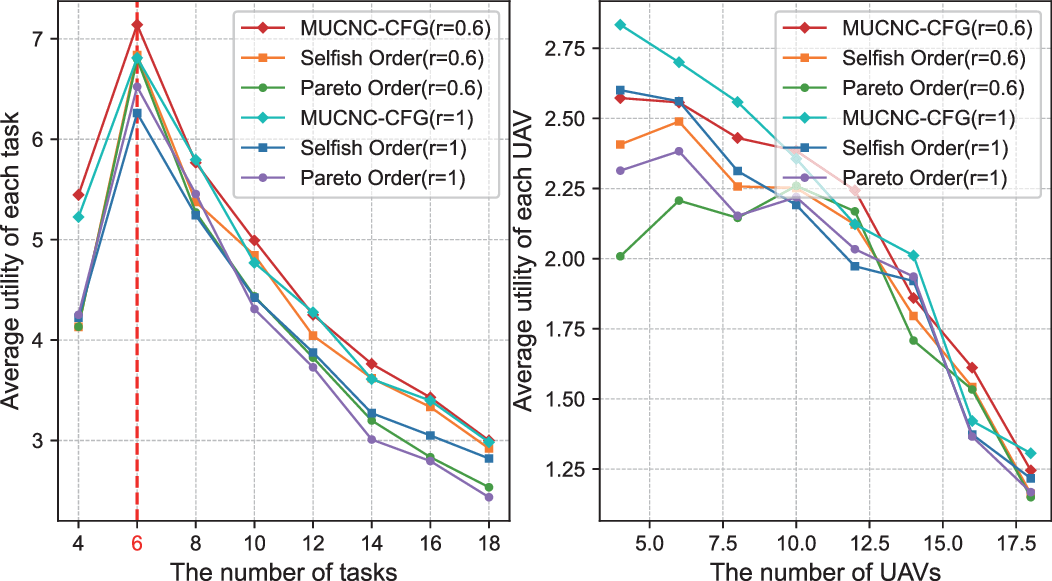}
\caption{The variation of average utility of each task and UAV with the task and UAV numbers for $N=20$ and $M=5$.}
\label{fig5}
\vspace{-0.4cm}
\end{figure}

In Figure~\ref{fig5}, we compare the average utility of each task when the number of UAVs is $N=20$. The experimental results show that, when the number of UAVs is fixed, the average utility of each task increases as the number of tasks increases and reaches its maximum at $N=6$, after which it decreases. This is because when the number of tasks slightly increases, UAVs cooperate to increase the utility of tasks. However, when the number of tasks increases significantly, UAVs may complete tasks individually, which reduces the average utility of each task.

Additionally, we compare the average utility of each UAV when the number of tasks is $M=4$. The experimental results indicate that, when the number of tasks is fixed, the average utility of each UAV decreases as the number of UAVs increases. When the number of UAVs is lower, our algorithm outperforms the other two in terms of the average utility of each UAV. Conversely, with higher UAV numbers, the three algorithms show no significant differences.

\subsubsection{Algorithm Convergence}

To verify the convergence of the algorithm, we studied the relationship between coalition overall utility and the number of iterations, as shown in Figure~\ref{fig6}. The experimental results indicate that our algorithm exhibits convergence during iterations, indicating that the coalition can reach a stable state. Although our algorithm does not have an advantage in terms of convergence speed compared to the other two algorithms, it eventually achieves higher overall utility. This optimization requires more iterations to find the best solution, but it ultimately results in higher overall utility.
\vspace{-0.2cm}
\begin{figure}[h]
\centering
\includegraphics[width=0.58\linewidth]{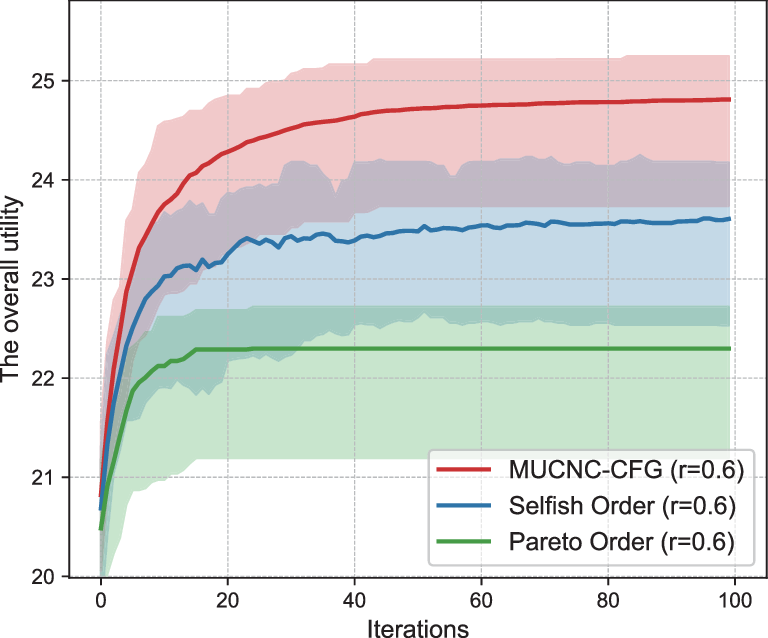}
\caption{The schematic diagram of the convergence of the three algorithms with $M=5$, $N=10$ and $r=0.6$.}
\label{fig6}
\end{figure}
\vspace{-0.4cm}

\section{Conlusion}
In this paper, in order to solve the tasks-driven multi-UAV coalition formation problem, we proposed a novel multi-UAV coalition for collaborative task completion model. Specifically, a revenue function based on coalition revenue threshold is firstly designed, then we used the Shapley value to distribute the utility of UAVs within the coalition based on the proposed model. Meanwhile, we designed a new Marginal Utility preference order based on this model, and proved that the CFG under this order has a Nash equilibrium solution. Finally, we proposed the MUCNC-CFG algorithm, which was able to achieve stable coalition structure within a limited number of iterations. Simulation results showed that the proposed algorithm can improve the overall utility of the coalition.




\begin{acks}
This work was supported by the National Natural Science Foundation
of China (No.61872313); the Special Innovation Fund
for Medical Innovation and Transformation – Clinical Translational
Research Project of Yangzhou University (No.AHYZUZHXM202103);
the Science and Technology on Near-Surface Detection Laboratory
(No.6142414220509); the Innovation and Entrepreneurship Training Program for college students in Jiangsu Province \seqsplit{(No.202311117059Z)}; the Startup Foundation for Introducing Talent of NUIST (No.2023r061). We thank several anonymous reviewers for their feedback on an earlier draft of this paper.
\end{acks}


\balance
\bibliographystyle{ACM-Reference-Format} 
\bibliography{sample}


\begin{thebibliography}{30}


\ifx \showCODEN    \undefined \def \showCODEN     #1{\unskip}     \fi
\ifx \showDOI      \undefined \def \showDOI       #1{#1}\fi
\ifx \showISBNx    \undefined \def \showISBNx     #1{\unskip}     \fi
\ifx \showISBNxiii \undefined \def \showISBNxiii  #1{\unskip}     \fi
\ifx \showISSN     \undefined \def \showISSN      #1{\unskip}     \fi
\ifx \showLCCN     \undefined \def \showLCCN      #1{\unskip}     \fi
\ifx \shownote     \undefined \def \shownote      #1{#1}          \fi
\ifx \showarticletitle \undefined \def \showarticletitle #1{#1}   \fi
\ifx \showURL      \undefined \def \showURL       {\relax}        \fi
\providecommand\bibfield[2]{#2}
\providecommand\bibinfo[2]{#2}
\providecommand\natexlab[1]{#1}
\providecommand\showeprint[2][]{arXiv:#2}

\bibitem[\protect\citeauthoryear{Apt and Witzel}{Apt and Witzel}{2009}]%
        {ANDREAS2009A}
\bibfield{author}{\bibinfo{person}{Krzysztof~R Apt} {and} \bibinfo{person}{Andreas Witzel}.} \bibinfo{year}{2009}\natexlab{}.
\newblock \showarticletitle{A generic approach to coalition formation}.
\newblock \bibinfo{journal}{\emph{International game theory review}} \bibinfo{volume}{11}, \bibinfo{number}{03} (\bibinfo{year}{2009}), \bibinfo{pages}{347--367}.
\newblock


\bibitem[\protect\citeauthoryear{Asheralieva and Niyato}{Asheralieva and Niyato}{2019}]%
        {asheralieva2019hierarchical}
\bibfield{author}{\bibinfo{person}{Alia Asheralieva} {and} \bibinfo{person}{Dusit Niyato}.} \bibinfo{year}{2019}\natexlab{}.
\newblock \showarticletitle{Hierarchical game-theoretic and reinforcement learning framework for computational offloading in UAV-enabled mobile edge computing networks with multiple service providers}.
\newblock \bibinfo{journal}{\emph{IEEE Internet of Things Journal}} \bibinfo{volume}{6}, \bibinfo{number}{5} (\bibinfo{year}{2019}), \bibinfo{pages}{8753--8769}.
\newblock


\bibitem[\protect\citeauthoryear{Barrot and Yokoo}{Barrot and Yokoo}{2019}]%
        {barrot2019stable}
\bibfield{author}{\bibinfo{person}{Nathana{\"e}l Barrot} {and} \bibinfo{person}{Makoto Yokoo}.} \bibinfo{year}{2019}\natexlab{}.
\newblock \showarticletitle{Stable and Envy-free Partitions in Hedonic Games.}. In \bibinfo{booktitle}{\emph{Proceedings of the International Joint Conference on Artificial Intelligence (IJCAI)}}. \bibinfo{pages}{67--73}.
\newblock


\bibitem[\protect\citeauthoryear{Boehmer, Bullinger, and Kerkmann}{Boehmer et~al\mbox{.}}{2023}]%
        {boehmer2023causes}
\bibfield{author}{\bibinfo{person}{Niclas Boehmer}, \bibinfo{person}{Martin Bullinger}, {and} \bibinfo{person}{Anna~Maria Kerkmann}.} \bibinfo{year}{2023}\natexlab{}.
\newblock \showarticletitle{Causes of stability in dynamic coalition formation}. In \bibinfo{booktitle}{\emph{Proceedings of the AAAI Conference on Artificial Intelligence}}, Vol.~\bibinfo{volume}{37}. \bibinfo{pages}{5499--5506}.
\newblock


\bibitem[\protect\citeauthoryear{Bogomolnaia and Jackson}{Bogomolnaia and Jackson}{2002}]%
        {2002The}
\bibfield{author}{\bibinfo{person}{Anna Bogomolnaia} {and} \bibinfo{person}{Matthew~O. Jackson}.} \bibinfo{year}{2002}\natexlab{}.
\newblock \showarticletitle{The Stability of Hedonic Coalition Structures}.
\newblock \bibinfo{journal}{\emph{Games \& Economic Behavior}} \bibinfo{volume}{38}, \bibinfo{number}{2} (\bibinfo{year}{2002}), \bibinfo{pages}{201--230}.
\newblock


\bibitem[\protect\citeauthoryear{Chalkiadakis, Elkind, and Wooldridge}{Chalkiadakis et~al\mbox{.}}{2022}]%
        {chalkiadakis2022computational}
\bibfield{author}{\bibinfo{person}{Georgios Chalkiadakis}, \bibinfo{person}{Edith Elkind}, {and} \bibinfo{person}{Michael Wooldridge}.} \bibinfo{year}{2022}\natexlab{}.
\newblock \bibinfo{booktitle}{\emph{Computational aspects of cooperative game theory}}.
\newblock \bibinfo{publisher}{Springer Nature}.
\newblock


\bibitem[\protect\citeauthoryear{Chen, Wu, Xu, Qi, Guan, Zhang, and Xue}{Chen et~al\mbox{.}}{2020}]%
        {2020Joint}
\bibfield{author}{\bibinfo{person}{Jiaxin Chen}, \bibinfo{person}{Qihui Wu}, \bibinfo{person}{Yuhua Xu}, \bibinfo{person}{Nan Qi}, \bibinfo{person}{Xin Guan}, \bibinfo{person}{Yuli Zhang}, {and} \bibinfo{person}{Zhen Xue}.} \bibinfo{year}{2020}\natexlab{}.
\newblock \showarticletitle{Joint task assignment and spectrum allocation in heterogeneous UAV communication networks: A coalition formation game-theoretic approach}.
\newblock \bibinfo{journal}{\emph{IEEE Transactions on Wireless Communications}} \bibinfo{volume}{20}, \bibinfo{number}{1} (\bibinfo{year}{2020}), \bibinfo{pages}{440--452}.
\newblock


\bibitem[\protect\citeauthoryear{Elkind, Fanelli, and Flammini}{Elkind et~al\mbox{.}}{2016}]%
        {elkind2016price}
\bibfield{author}{\bibinfo{person}{Edith Elkind}, \bibinfo{person}{Angelo Fanelli}, {and} \bibinfo{person}{Michele Flammini}.} \bibinfo{year}{2016}\natexlab{}.
\newblock \showarticletitle{Price of Pareto optimality in hedonic games}. In \bibinfo{booktitle}{\emph{Proceedings of the AAAI Conference on Artificial Intelligence}}, Vol.~\bibinfo{volume}{30}.
\newblock


\bibitem[\protect\citeauthoryear{Gupta, Jain, and Vaszkun}{Gupta et~al\mbox{.}}{2015}]%
        {gupta2015survey}
\bibfield{author}{\bibinfo{person}{Lav Gupta}, \bibinfo{person}{Raj Jain}, {and} \bibinfo{person}{Gabor Vaszkun}.} \bibinfo{year}{2015}\natexlab{}.
\newblock \showarticletitle{Survey of important issues in UAV communication networks}.
\newblock \bibinfo{journal}{\emph{IEEE communications surveys \& tutorials}} \bibinfo{volume}{18}, \bibinfo{number}{2} (\bibinfo{year}{2015}), \bibinfo{pages}{1123--1152}.
\newblock


\bibitem[\protect\citeauthoryear{Han, Yang, Bilal, Wang, Krichen, Alsadhan, and Ge}{Han et~al\mbox{.}}{2023}]%
        {han2023smart}
\bibfield{author}{\bibinfo{person}{Zhaoyang Han}, \bibinfo{person}{Yaoqi Yang}, \bibinfo{person}{Muhammad Bilal}, \bibinfo{person}{Weizheng Wang}, \bibinfo{person}{Moez Krichen}, \bibinfo{person}{Abeer~Abdullah Alsadhan}, {and} \bibinfo{person}{Chunpeng Ge}.} \bibinfo{year}{2023}\natexlab{}.
\newblock \showarticletitle{Smart Optimization Solution for Channel Access Attack Defense under UAV-aided Heterogeneous Network}.
\newblock \bibinfo{journal}{\emph{IEEE Internet of Things Journal}} (\bibinfo{year}{2023}).
\newblock


\bibitem[\protect\citeauthoryear{Huang, Qi, Huang, Jia, Wu, Yao, and Wang}{Huang et~al\mbox{.}}{2022}]%
        {huang2022connectivity}
\bibfield{author}{\bibinfo{person}{Yeting Huang}, \bibinfo{person}{Nan Qi}, \bibinfo{person}{Zanqi Huang}, \bibinfo{person}{Luliang Jia}, \bibinfo{person}{Qihui Wu}, \bibinfo{person}{Rugui Yao}, {and} \bibinfo{person}{Wenjing Wang}.} \bibinfo{year}{2022}\natexlab{}.
\newblock \showarticletitle{Connectivity guarantee within UAV cluster: A graph coalition formation game approach}.
\newblock \bibinfo{journal}{\emph{IEEE Open Journal of the Communications Society}}  \bibinfo{volume}{4} (\bibinfo{year}{2022}), \bibinfo{pages}{79--90}.
\newblock


\bibitem[\protect\citeauthoryear{L{\~a}, Chew, and Soong}{L{\~a} et~al\mbox{.}}{2016}]%
        {2016Potential}
\bibfield{author}{\bibinfo{person}{Quang~Duy L{\~a}}, \bibinfo{person}{Yong~Huat Chew}, {and} \bibinfo{person}{Boon-Hee Soong}.} \bibinfo{year}{2016}\natexlab{}.
\newblock \showarticletitle{Potential game theory}.
\newblock \bibinfo{journal}{\emph{Cham: Springer International Publishing}} (\bibinfo{year}{2016}).
\newblock


\bibitem[\protect\citeauthoryear{Lucchetti, Moretti, and Rea}{Lucchetti et~al\mbox{.}}{2022}]%
        {lucchetti2022coalition}
\bibfield{author}{\bibinfo{person}{Roberto Lucchetti}, \bibinfo{person}{Stefano Moretti}, {and} \bibinfo{person}{Tommaso Rea}.} \bibinfo{year}{2022}\natexlab{}.
\newblock \showarticletitle{Coalition formation games and social ranking solutions}. In \bibinfo{booktitle}{\emph{Proceedings of the International Conference on Autonomous Agents and Multiagent Systems (AAMAS)}}.
\newblock


\bibitem[\protect\citeauthoryear{Mutzari, Gan, and Kraus}{Mutzari et~al\mbox{.}}{2021}]%
        {mutzari2021coalition}
\bibfield{author}{\bibinfo{person}{Dolev Mutzari}, \bibinfo{person}{Jiarui Gan}, {and} \bibinfo{person}{Sarit Kraus}.} \bibinfo{year}{2021}\natexlab{}.
\newblock \showarticletitle{Coalition formation in multi-defender security games}. In \bibinfo{booktitle}{\emph{Proceedings of the AAAI Conference on Artificial Intelligence}}, Vol.~\bibinfo{volume}{35}. \bibinfo{pages}{5603--5610}.
\newblock


\bibitem[\protect\citeauthoryear{Ng, Lim, Dai, Xiong, Huang, Niyato, Hua, Leung, and Miao}{Ng et~al\mbox{.}}{2020}]%
        {ng2020joint}
\bibfield{author}{\bibinfo{person}{Jer~Shyuan Ng}, \bibinfo{person}{Wei Yang~Bryan Lim}, \bibinfo{person}{Hong-Ning Dai}, \bibinfo{person}{Zehui Xiong}, \bibinfo{person}{Jianqiang Huang}, \bibinfo{person}{Dusit Niyato}, \bibinfo{person}{Xian-Sheng Hua}, \bibinfo{person}{Cyril Leung}, {and} \bibinfo{person}{Chunyan Miao}.} \bibinfo{year}{2020}\natexlab{}.
\newblock \showarticletitle{Joint auction-coalition formation framework for communication-efficient federated learning in UAV-enabled internet of vehicles}.
\newblock \bibinfo{journal}{\emph{IEEE Transactions on Intelligent Transportation Systems}} \bibinfo{volume}{22}, \bibinfo{number}{4} (\bibinfo{year}{2020}), \bibinfo{pages}{2326--2344}.
\newblock


\bibitem[\protect\citeauthoryear{Qi, Huang, Zhou, Shi, Wu, and Xiao}{Qi et~al\mbox{.}}{2022}]%
        {qi2022task}
\bibfield{author}{\bibinfo{person}{Nan Qi}, \bibinfo{person}{Zanqi Huang}, \bibinfo{person}{Fuhui Zhou}, \bibinfo{person}{Qingjiang Shi}, \bibinfo{person}{Qihui Wu}, {and} \bibinfo{person}{Ming Xiao}.} \bibinfo{year}{2022}\natexlab{}.
\newblock \showarticletitle{A task-driven sequential overlapping coalition formation game for resource allocation in heterogeneous UAV networks}.
\newblock \bibinfo{journal}{\emph{IEEE Transactions on Mobile Computing}} (\bibinfo{year}{2022}).
\newblock


\bibitem[\protect\citeauthoryear{Rahwan and Jennings}{Rahwan and Jennings}{2008}]%
        {rahwan2008coalition}
\bibfield{author}{\bibinfo{person}{Talal Rahwan} {and} \bibinfo{person}{Nick Jennings}.} \bibinfo{year}{2008}\natexlab{}.
\newblock \showarticletitle{Coalition structure generation: Dynamic programming meets anytime optimisation}.
\newblock  (\bibinfo{year}{2008}).
\newblock


\bibitem[\protect\citeauthoryear{Saad, Han, Basar, Debbah, and Hjorungnes}{Saad et~al\mbox{.}}{2009}]%
        {saad2009selfish}
\bibfield{author}{\bibinfo{person}{Walid Saad}, \bibinfo{person}{Zhu Han}, \bibinfo{person}{Tamer Basar}, \bibinfo{person}{M{\'e}rouane Debbah}, {and} \bibinfo{person}{Are Hjorungnes}.} \bibinfo{year}{2009}\natexlab{}.
\newblock \showarticletitle{A selfish approach to coalition formation among unmanned air vehicles in wireless networks}. In \bibinfo{booktitle}{\emph{Proceedings of the International Conference on Game Theory for Networks}}. \bibinfo{pages}{259--267}.
\newblock


\bibitem[\protect\citeauthoryear{Shakeri, Al-Garadi, Badawy, Mohamed, Khattab, Al-Ali, Harras, and Guizani}{Shakeri et~al\mbox{.}}{2019}]%
        {shakeri2019design}
\bibfield{author}{\bibinfo{person}{Reza Shakeri}, \bibinfo{person}{Mohammed~Ali Al-Garadi}, \bibinfo{person}{Ahmed Badawy}, \bibinfo{person}{Amr Mohamed}, \bibinfo{person}{Tamer Khattab}, \bibinfo{person}{Abdulla~Khalid Al-Ali}, \bibinfo{person}{Khaled~A Harras}, {and} \bibinfo{person}{Mohsen Guizani}.} \bibinfo{year}{2019}\natexlab{}.
\newblock \showarticletitle{Design challenges of multi-UAV systems in cyber-physical applications: A comprehensive survey and future directions}.
\newblock \bibinfo{journal}{\emph{IEEE Communications Surveys \& Tutorials}} \bibinfo{volume}{21}, \bibinfo{number}{4} (\bibinfo{year}{2019}), \bibinfo{pages}{3340--3385}.
\newblock


\bibitem[\protect\citeauthoryear{Shapley}{Shapley}{1953}]%
        {1953A}
\bibfield{author}{\bibinfo{person}{L.~S. Shapley}.} \bibinfo{year}{1953}\natexlab{}.
\newblock \showarticletitle{A value for $n$-person games}.
\newblock \bibinfo{journal}{\emph{annals of mathematical studies}} (\bibinfo{year}{1953}).
\newblock


\bibitem[\protect\citeauthoryear{Shapley}{Shapley}{1996}]%
        {1996Potential}
\bibfield{author}{\bibinfo{person}{Monderer Lloyd~S. Shapley}.} \bibinfo{year}{1996}\natexlab{}.
\newblock \showarticletitle{Potential Games}.
\newblock \bibinfo{journal}{\emph{Games and Economic Behavior}} (\bibinfo{year}{1996}).
\newblock


\bibitem[\protect\citeauthoryear{Shehory and Kraus}{Shehory and Kraus}{1998}]%
        {shehory1998methods}
\bibfield{author}{\bibinfo{person}{Onn Shehory} {and} \bibinfo{person}{Sarit Kraus}.} \bibinfo{year}{1998}\natexlab{}.
\newblock \showarticletitle{Methods for task allocation via agent coalition formation}.
\newblock \bibinfo{journal}{\emph{Artificial intelligence}} \bibinfo{volume}{101}, \bibinfo{number}{1-2} (\bibinfo{year}{1998}), \bibinfo{pages}{165--200}.
\newblock


\bibitem[\protect\citeauthoryear{Tang, Duan, and Lao}{Tang et~al\mbox{.}}{2023}]%
        {tang2023swarm}
\bibfield{author}{\bibinfo{person}{Jun Tang}, \bibinfo{person}{Haibin Duan}, {and} \bibinfo{person}{Songyang Lao}.} \bibinfo{year}{2023}\natexlab{}.
\newblock \showarticletitle{Swarm intelligence algorithms for multiple unmanned aerial vehicles collaboration: A comprehensive review}.
\newblock \bibinfo{journal}{\emph{Artificial Intelligence Review}} \bibinfo{volume}{56}, \bibinfo{number}{5} (\bibinfo{year}{2023}), \bibinfo{pages}{4295--4327}.
\newblock


\bibitem[\protect\citeauthoryear{Vig and Adams}{Vig and Adams}{2006}]%
        {vig2006multi}
\bibfield{author}{\bibinfo{person}{Lovekesh Vig} {and} \bibinfo{person}{Julie~A Adams}.} \bibinfo{year}{2006}\natexlab{}.
\newblock \showarticletitle{Multi-robot coalition formation}.
\newblock \bibinfo{journal}{\emph{IEEE transactions on robotics}} \bibinfo{volume}{22}, \bibinfo{number}{4} (\bibinfo{year}{2006}), \bibinfo{pages}{637--649}.
\newblock


\bibitem[\protect\citeauthoryear{Wang, Song, Han, and Jiao}{Wang et~al\mbox{.}}{2013}]%
        {wang2013dynamic}
\bibfield{author}{\bibinfo{person}{Tianyu Wang}, \bibinfo{person}{Lingyang Song}, \bibinfo{person}{Zhu Han}, {and} \bibinfo{person}{Bingli Jiao}.} \bibinfo{year}{2013}\natexlab{}.
\newblock \showarticletitle{Dynamic popular content distribution in vehicular networks using coalition formation games}.
\newblock \bibinfo{journal}{\emph{IEEE Journal on Selected Areas in Communications}} \bibinfo{volume}{31}, \bibinfo{number}{9} (\bibinfo{year}{2013}), \bibinfo{pages}{538--547}.
\newblock


\bibitem[\protect\citeauthoryear{Wu and Ramchurn}{Wu and Ramchurn}{2020}]%
        {wu2020monte}
\bibfield{author}{\bibinfo{person}{Feng Wu} {and} \bibinfo{person}{Sarvapali~D Ramchurn}.} \bibinfo{year}{2020}\natexlab{}.
\newblock \showarticletitle{Monte-Carlo tree search for scalable coalition formation}. In \bibinfo{booktitle}{\emph{Proceedings of the International Joint Conference on Artificial Intelligence (IJCAI)}}. \bibinfo{pages}{407--413}.
\newblock


\bibitem[\protect\citeauthoryear{Wu and Shang}{Wu and Shang}{2020}]%
        {WU2020208}
\bibfield{author}{\bibinfo{person}{Han Wu} {and} \bibinfo{person}{Huiliang Shang}.} \bibinfo{year}{2020}\natexlab{}.
\newblock \showarticletitle{Potential game for dynamic task allocation in multi-agent system}.
\newblock \bibinfo{journal}{\emph{ISA transactions}}  \bibinfo{volume}{102} (\bibinfo{year}{2020}), \bibinfo{pages}{208--220}.
\newblock


\bibitem[\protect\citeauthoryear{Xiong, Zheng, Ruan, Wang, Tang, Dong, and Li}{Xiong et~al\mbox{.}}{2020}]%
        {xiong2020energy}
\bibfield{author}{\bibinfo{person}{Fei Xiong}, \bibinfo{person}{Hao Zheng}, \bibinfo{person}{Lang Ruan}, \bibinfo{person}{Hai Wang}, \bibinfo{person}{Lijuan Tang}, \bibinfo{person}{Xu Dong}, {and} \bibinfo{person}{Aijing Li}.} \bibinfo{year}{2020}\natexlab{}.
\newblock \showarticletitle{Energy-saving data aggregation for multi-UAV system}.
\newblock \bibinfo{journal}{\emph{IEEE Transactions on Vehicular Technology}} \bibinfo{volume}{69}, \bibinfo{number}{8} (\bibinfo{year}{2020}), \bibinfo{pages}{9002--9016}.
\newblock


\bibitem[\protect\citeauthoryear{Yang and Luo}{Yang and Luo}{2007}]%
        {yang2007coalition}
\bibfield{author}{\bibinfo{person}{Jingan Yang} {and} \bibinfo{person}{Zhenghu Luo}.} \bibinfo{year}{2007}\natexlab{}.
\newblock \showarticletitle{Coalition formation mechanism in multi-agent systems based on genetic algorithms}.
\newblock \bibinfo{journal}{\emph{Applied Soft Computing}} \bibinfo{volume}{7}, \bibinfo{number}{2} (\bibinfo{year}{2007}), \bibinfo{pages}{561--568}.
\newblock


\bibitem[\protect\citeauthoryear{Zhang, Xu, Wu, Luo, Xu, Chen, Anpalagan, and Zhang}{Zhang et~al\mbox{.}}{2018}]%
        {zhang2018context}
\bibfield{author}{\bibinfo{person}{Yuli Zhang}, \bibinfo{person}{Yuhua Xu}, \bibinfo{person}{Qihui Wu}, \bibinfo{person}{Yunpeng Luo}, \bibinfo{person}{Yitao Xu}, \bibinfo{person}{Xueqiang Chen}, \bibinfo{person}{Alagan Anpalagan}, {and} \bibinfo{person}{Daoqiang Zhang}.} \bibinfo{year}{2018}\natexlab{}.
\newblock \showarticletitle{Context awareness group buying in D2D networks: A coalition formation game-theoretic approach}.
\newblock \bibinfo{journal}{\emph{IEEE Transactions on Vehicular Technology}} \bibinfo{volume}{67}, \bibinfo{number}{12} (\bibinfo{year}{2018}), \bibinfo{pages}{12259--12272}.
\newblock


\end{thebibliography}


\end{document}